\documentclass[submission,copyright,creativecommons]{eptcs}
\providecommand{\event}{SOS 2007} 
\usepackage{breakurl}             
\usepackage{underscore}           
\usepackage{amssymb,amsmath,amsfonts}
\usepackage{graphicx}
\usepackage{amsmath}
\usepackage{amsfonts}
\usepackage{amssymb}%
\providecommand{\U}[1]{\protect\rule{.1in}{.1in}}
\newtheorem{theorem}{Theorem}

\newtheorem{defi}[theorem]{Definition}

\newtheorem{idea memo}[theorem]{Idea Memo}

\newtheorem{remark}[theorem]{Remark}

\newenvironment{proof}[1][Proof]{\textbf{#1.}}{\ \rule{0.5em}{0.5em}}

\title{Quantum Walk and Dressed Photon}
\author{Misa Hamano
\institute{Nagahama Institute of Bio-Science and Technology\\ Nagahama, Japan}
\email{misa.flute617klfx@gmail.com}
\and
Hayato Saigo
\institute{Nagahama Institute of Bio-Science and Technology\\ Nagahama, Japan}
\email{harmoniahayato@gmail.com}
}
\def\titlerunning{Quantum Walk and Dressed Photon}
\def\authorrunning{M. Hamano \& H. Saigo}
\begin{document}
\maketitle

\begin{abstract}
A physical model called dressed photon, a composite system of photons and excitation of electrons in the nano-particles, is effectively used in the realm of near-field optics. Many interesting behaviors of dressed photons are known, especially the rapid energy transfer and the accumulation to singular points, e.g., points with strong dissipation. We propose a new modelling of dressed photons based on quantum walks, especially Grover walks on semi-infinite graphs which we call jellyfish graphs, and show a universal accumulation phenomena around the point with strong dissipation.
\end{abstract}

\section{Introduction}
The subject of this paper is to analyze the behavior of dressed photons \cite{O} using quantum walks. In this section, we briefly review what quantum walks and 
dressed photons are, and then describes the structure of this paper.

Quantum walks are, in a nutshell, "quantum version of the random walk" (see \cite{A}, for example). More specifically, we generally consider the complex "probability amplitudes" and the unitary evolution of them. Here, the probability amplitude is a quantity which the probability can be obtained by squaring its absolute value, and the unitary evolution is the time evolution by a reversible linear mapping which keeps "sum of the squares of the absolute values of the probability amplitudes". 

In general, the sum of the probabilities is normalized to be 1, but it is convenient not to normalize such a value when the quantum beam is constantly incident and emitted as the phenomena focused in this paper, and  
we simply call them "amplitudes". Where the square of the absolute value of amplitudes is high, the probability is relatively high.

Quantum walks seem to have an analogous definition to random walks, but in reality there are major differences in nature. For example, the average travel distance of a random walk is generally proportional to the "square root" of time, while that of a quantum walk is proportional to the time itself. Quantum walks are "rapid" compared to random walks. Because of this nature, quantum walks are considered as the basis for more efficient search and computing. 

The quantum walk can also be seen as a toy model of the quantum field (see \cite{D}, for example). As is well known, quantum phenomena should be considered as the behaviors of fields as system with spatial and temporal extent, but in measurement they appear as local and particle-like phenomena. Mathematical models for such system are called "quantum fields", which is defined as a system of quantities spread from place to place, from which the expectation values of the measurements of those quantities are determined (for each interaction with the environment). Such a correspondence from quantity to expectation value is called "state". It is known that any state in this meaning can generally be expressed as a vector of Hilbert space. In short, a quantum field is considered as a system composed of algebras of physical quantities and states (or "local states" \cite{OOS}) on them. 

Quantum fields defined in this way are the most central concept in modern physics, but their analysis requires sophisticated mathematical techniques. It is also known that, apart from free fields, no such non-trivial mathematical model can be constucted that satisfies the "seemingly natural" axioms.

Therefore, 
it is interesting to 
use "Quantum walk on the graph" in which the amplitude is defined on each arrow at each time as a simple model, breaking the basic properties assumed in the conventional quantum field theory.
In particular, this modeling is suitable for modeling the light as a quantum field that interacts with nanoparticles called "dressed photon".

The dressed photon \cite{O} is a useful concept to consider the behavior of a quantum field formed by combining a photon field and an electron (excitation of) field of a nanoparticle. It has an overwhelmingly larger number of modes (momentum and energy) than the incoming light itself. Of these modes, the momentum and energy "relatively high" modes manifest themselves as "a small grain of light" localized to the "interval" of the nanoparticle. This is the dressed photon ("relatively low" modes are thought to function as a kind of heat bath). Simply put, the dressed photons are a new kind of quantum created "between" nanoparticles when light combines with the nanoparticle system.


The main subject of this paper is to analyze the behavior of dressed photons using "Grover walks on jellyfish graphs" as a first step, or a toy model, toward such a theoretical foundation. After the preliminaries (section 2), modelling based on concepts such as jellyfish graphs and Grover walks on them are introduced in section 3. In section 4, the main theorem of the present paper are introduced based on these preparations. In Section 5, we discuss the physical meaning of the theorem. 

\section{Preliminaries}

\subsection{Directed Graph}

Simply put, a directed graph is a figure consisting of "vertices" and "arrows" connecting them. Mathematically correct definition is as follows.

\begin{defi}[Directed Graph]
A directed graph (or a digraph) $G$ is a quadruple  
$(V_G,A_G,o_G,t_G)$ composed of a set $V_G$, a set $A_G$, 
a mapping $o_G:A_G\longrightarrow V_G$ and 
a mapping $t_G:A_G\longrightarrow V_G$. 
The elements of $V_G$ and of $A_G$ are called vertices and arrows, respectively. 
For an arrow $a$, $o_G(a)$ is called the origin of $a$ and $t_G(a)$ is called the target of $a$.
\end{defi}


With regard to directed graphs, the concept of "path" is important.

\begin{defi}
A path in a directed graph $G$ is a (finite) sequence of arrows $a_1,a_2,...,a_n$ in $G$ such that $o_G(a_{i+1})=t_G(a_i)$ holds for any $i$. $o(a_1)$ is called the origin of the path and $t(a_n)$ is called the target of the path. 
\end{defi}

In this paper, "connected" directed graphs play an important role. The connectivity of a graph is defined as follows using the concept of "path".

\begin{defi}
A directed graph is called connected if for any two vertices $v,v'$ there is some path whose origin is $v$ and whose target is $v'$. \footnote{Although the notion of connectivity defined here is usually called "strongly connected", no confusion will be occurred by this abuse of the term for the kinds of graphs treated in the present paper.}.
\end{defi}

\subsection{Symmetric Simple Directed Graph}
This paper deals with "symmetric simple directed graphs". The symmetry and the simplicity of directed graphs are defined as follows:



\begin{defi}
A directed graph $G=(V_G,A_G,o_G,t_G)$ is called symmetric if there exist a mapping $\bar{(\: )}:A_G \longrightarrow A_G$ such that 
$o_G(\bar{a})=t_G(a), t_G(\bar{a})=o_G(a)$ holds.
\end{defi}


\begin{defi}
A directed graph $G=(V_G,A_G,o_G,t_G)$ is called simple 
if for any $v,v'\in V_G$ there exists at most one 
$a\in A_G $ such that 
$v=o_G(a), v'=t_G(a)$ hold and for any $a\in A_G $,
$o_G(a)\neq t_G(a)$ holds.
\end{defi}


\textbf{Hereafter, in this paper, "a graph" simply means "a symmetric simple directed graph".} 
Such graph can also be considered as 
"a directed graph obtained from simple undirected graph by replacing each edge by the pair of opposite arrows".  The notion of simple undirescted graph can be defined as follows:


\begin{defi}
A pair of sets 
$\Gamma =(V_{\Gamma},E_{\Gamma})$ is called a simple undirected graph if any of elements of $E_{\Gamma}$ is a set composed of two elements of $V_{\Gamma}$. An element of $V_{\Gamma}$ is called a vertex of $\Gamma$, and an element of $E_{\Gamma}$ is called an edge of $\Gamma$. An edge $e$ is called incident to a vertex $v$ if $v\in e$. The number of edges incident to a vertex $v$ is called the degree of $v$ and denoted as $deg(v)$.
\end{defi}






\textbf{Since the graphs (symmetric simple directed graphs) can be considered as essentially the same as simple undirected graphs, we apply the terms originally defined for the latter to the former. 
Also, we often omit the indices indicating the name of graphs, especially for the targets and origins.  }

\section{Modelling}

\subsection{Jellyfish Graphs}

The main focus of this paper is quantum walks on a kind of graphs, called ``jellyfish graphs'' defined below, which is useful for constructing a toy model of the dressed photon phenomenon. 
In short, a jellyfish graph is a graph composed of a finite connected graph and a finite number of "half-line" attached to it. More precisely, it is defined as follows.

\begin{defi}
A graph $G$ is called a jellyfish graph if it is the union of a finite number of graphs 
\[G^{(0)},l^{(1)},l^{(2)},...\]
which satisfy the following conditions:
\begin{itemize}
    \item $G^{(0)}$ is finite connected graph, i.e., connected graph composed of a finite number of vertices and arrows. For simplicity, we identify the set of vertices with the set $\{1,2,3,...,n\}$, where $n$ denotes the number of vertices.
    
    \item Each $l^{(i)}$ is a half-line graph, i.e., a connected graph such that degree of any vertex $v$ is $2$ except for one vertex called the endvertex of $l^{(i)}$ which is the unique common vertex of $l^{(i)}$ and $G^{(0)}$. 
    
\end{itemize}
\end{defi}

\subsection{Quantum Walks on Jellyfish Graphs}

A toy model of dressed photon phenomena can be formulated as a quantum walk on a jellyfish graph, based on the picture that the light is injected to a nanoparticle system from a distance, and it becomes a dressed photon between the nanoparticles for a while, and then emitted as light from there.

Although the term "jellyfish graph" itself is introduced in this paper, 
Feldman and Hillery already considered the discrete time quantum walk on this kind of graph\cite{FH} and continuous time quantum walk had been considered by Farhi and Gutmann\cite{FG}. Higuchi and Segawa \cite{HS} also study this type of graph and showed the following theorem (We modify the statement with our terminology).

\begin{theorem}
Consider a quantum walk on jellyfish graph being free on each half-line $l^{(i)}$, i.e., a quantum walk satisfying the condition that for any arrow $a$ in $l^{(i)}$ and any arrow $a'$ in $l^{(i)}$ such that $o(a')=t(a), t(a')\neq o(a)$
\[
\psi_{t+1}(a')=\psi_{t}(a)
\]
holds, where $\psi_{t}(a)$ and $\psi_{t+1}(a')$ denotes the amplitude of $a$ and $a'$ at the time $t$ and $t+1$, respectively.

If such quantum walk starts from the initial condition with
\begin{itemize}
    \item constant amplitudes ${\alpha}^{(i)}$ on the arrows in each half-line $l^{(i)}$ directed to the end vertex
    \item amplitude $0$ on other arrows
\end{itemize}
then the amplitude ${\psi}_{t}(a)$ converges to the limit amplitude ${\psi}_{\infty}(a)$ when $t$ tends to infinity.
 \end{theorem}
 
\subsection{Grover Walks on Jellyfish Graphs} 

The theorem above holds for general quantum walks. If we focus on the typical concrete quantum walks called Grover walks, we obtain more detailed results. Here, Grover walk means the quantum walk such that the amplitudes in time $t$ of arrows $a,a',a'',...$ with the common target determines the ones in time $t+1$ of $\bar{a},\bar{a'},\bar{a''},...$, and the relation ship between these amplitudes are given by

\[
\begin{bmatrix}
\beta \\
{\beta}' \\
{\beta}'' \\
\vdots \\
\end{bmatrix}
=
\begin{bmatrix}
2/r-1 & 2/r & \cdots & 2/r \\
2/r   & 2/r-1 & \cdots & 2/r \\
\vdots  & \vdots & \ddots & \vdots \\
2/r   & 2/r &  \cdots & 2/r-1
\end{bmatrix}
\begin{bmatrix}
\alpha \\
{\alpha}' \\
{\alpha}'' \\
\vdots \\
\end{bmatrix}
\]
where $r$ denotes the degree of the common vertex. 



Physically speaking, the idea of a Grover walk corresponds to the idea of "scattering in the absence of potential" but it is also considered to be the first step in considering the energy transport of dressed photons (In the sense that we first consider the case where there is no potential, and then make corrections as necessary).

Higuchi and Segawa \cite{HS} showed the following fundamental theorem about the limit amplitudes of the Grover walks on jellyfish graphs.

\begin{theorem}
Let ${\alpha}^{(i)}$ denotes the constant amplitude of $a$ in $l^{(i)}$ and ${\beta}^{(i)}={\psi}_{\infty}(\bar{a})$ denotes the limit amplitude of $a$, which is shown to be constant for each $l^{(i)}$. The relation between these amplitudes are given as follows:

\[
\begin{bmatrix}
{\beta}^{(1)} \\
{\beta}^{(2)} \\
{\beta}^{(3)} \\
\vdots \\
\end{bmatrix}
=
\begin{bmatrix}
2/r-1 & 2/r & \cdots & 2/r \\
2/r   & 2/r-1 & \cdots & 2/r \\
\vdots  & \vdots & \ddots & \vdots \\
2/r   & 2/r &  \cdots & 2/r-1
\end{bmatrix}
\begin{bmatrix}
{\alpha}^{(1)}\\
{\alpha}^{(2)} \\
{\alpha}^{(3)}\\
\vdots \\
\end{bmatrix}
\]
\end{theorem}

Simply put, this theorem means that when viewed "from far away and after enough time", the nanoparticle system itself appears to be as if a "vertex" of the Grover walk, with the half-line viewed as if the arrows incident to the vertex.

In \cite{HS}, the following "Kirchhoff type" theorem is also proved:

\begin{theorem}
Let $J(a)$ be the quantity defined as 
\[J(a):=\psi_{\infty}(a)-ave(\alpha^{(1)},\alpha^{(2)},\dotsb ,\alpha^{(m)}),
\]
where $ave(\alpha^{(1)},\alpha^{(2)},\dotsb ,\alpha^{(m)})$ denotes the average of $\alpha^{(1)},\alpha^{(2)},\dotsb ,\alpha^{(m)}$.
The following equations hold:
\begin{itemize}
    \item $J(a)+J(\bar{a})=0.$
    
    \item For any $v\in V_G$, $\displaystyle \sum_{a\in A_G, t(a)=v} J(a)=0.$
\end{itemize}

\end{theorem}

This is nothing but the "Kirchhoff's current law" for the quantity $J(a)$. Similarly, a law corresponding to the  "Kirchhoff's voltage law" is also shown, but it is omitted here. It is important to note that these law on $J(a)$ allows us to calculate the limit amplitudes.

In the next section, we use the theorem introduced in this section to describe a new theorem for the limit amplitudes of quantum walks on jellyfish graphs.

\section{Result}

We made the following speculation based on the physical insight of the dressed photons and the analogy with the quantum walk:

\begin{quote}
Consider the Grover walk on a jellyfish graph with the initial amplitudes ${\alpha}^{(1)},{\alpha}^{(2)},...{\alpha}^{(m-1)}$ on the arrows in $l^{(1)}....l^{(m-1)}$ directed to $G^{(0)}$. Then the difference between the sum of the square of the absolute value of limit amplitudes on arrows into a vertex and the one on arrows out of it, which will correspond to the (not normalized) probability ``around the vertex'', will be maximized at the vertex with the maximum sum of the amplitudes of arrow out of it into the exterior of $G^{(0)}$. 
\end{quote}

Based on this speculation, we obtain the following simple theorem. The physical meaning of quantity $P^{0}(v)$ defined in the statement is the net amplitude of dressed photon around the vertex $v$ (For the detailed physical interpretation of the theorem, see the next section) .

\begin{theorem}
Let $P^{(0)}(v)$ be the quantity defined as 



\[
\displaystyle P^{(0)}(v)=\sum_{t(a)=v,a\in G^{(0)}}|\psi(a)|^2-\sum_{t(a)=v,a\in G^{(0)}}|\psi(\bar{a})|^2.
\]
For any $v\in V_{G^{(0)}}$

\[
P^{(0)}(v)=4ave(\alpha^{(1)},
\alpha^{(2)},\dotsb,\alpha^{(m)})J^{out}(v)\\
\]
holds, where $\alpha^{(1)},\alpha^{(2)},\dotsb,\alpha^{(m)}$ are real numbers 
and $J^{out}(v)$ denotes the quantity defined as
\[
J^{out}(v)=\sum_{o(a)=v,a\in {G/G^{(0)}}}J(a)
\]
which is equal to $\displaystyle \sum_{t(a)=v,a\in G^{(0)}}J(a)$ by Kirchhoff's law.
\end{theorem}
\begin{proof}

Let us denote $ave(\alpha^{(1)},
\alpha^{(2)},\dotsb,\alpha^{(m)})$ as $ave$.

$\displaystyle 
P^{(0)}(v)=\sum_{t(a)=v,a\in G^{(0)}}|\psi(a)|^2-\sum_{t(a)=v,a\in G^{(0)}}|\psi(\bar{a})|^2$\\
$\displaystyle
\ \ \ \ \ \ \ \  =\sum_{t(a)=v,a\in G^{(0)}}\{|J(a)+ave|^2-|-J(a)+ave|^2\}$\\
$\displaystyle
\ \ \ \ \ \ \ \ \ =\sum_{t(a)=v,a\in G^{(0)}}\{(J(a)^2+2J(a)ave+ave^2)-(J(a)^2-2J(a)ave+ave^2)\}$\\
$\displaystyle
\ \ \ \ \ \ \ \ =\sum_{t(a)=v,a\in G^{(0)}}\{J(a)^2+2J(a)ave+ave^2-J(a)^2+2J(a)ave-ave^2\}$\\
$\displaystyle
\ \ \ \ \ \ \ \ =\sum_{t(a)=v,a\in G^{(0)}}4J(a)ave$\\
$\displaystyle
\ \ \ \ \ \ \ \ =4ave\sum_{t(a)=v,a\in G^{(0)}}J(a)$\\
\ \ \ \ \ \ \ \ $=4ave J^{out}(v)=4ave(\alpha^{(1)},
\alpha^{(2)},\dotsb,\alpha^{(m)})J^{out}(v)$.
\end{proof}

\begin{remark}
We can easily generalize the theorem and proof above for any complex amplitudes.(Omitted here for simplicity.)

\end{remark}



\section{Discussion}

Let us consider the physical meaning of the theorem above.

Quantum walks on jellyfish graphs are considered to correspond to a physical model in which light enters a nanoparticle system from a distance, becomes a dressed photon for a while, and then emitted into far field. 

If this modeling is appropriate, then the above theorem means:
\begin{quote}
Dressed photons cluster around vertices that emit a large amount of outflow into the distance.
\end{quote}
In other words, the dressed photons gather "in a self-compensating manner" where there is a large "dissipation"(into the far field).

In fact, many phenomena being coherent to the result above have been confirmed experimentally. For example, a fiber probe used as a "generator" of the dressed photons receives light from its root and emits light to the external environment mostly at its surface, especially at its tip, where it is experimentally known that the density of dressed photons are known highest\cite{PYKFO}. 

Similarly, when visible light (not ultraviolet light,) enters the surface of a substance in contact with a gas such as chlorine or oxygen, a chemical reaction that cannot occur at the original frequency due to low energy occurs at the tip of a nanostructure similar to a fiber probe, and the sharp part is polished autonomously\cite{YNNO}. It has also been put to practical use as a nanometer scale polishing technology. Another example is a device called "optical nanofountain" that uses nanoparticles of different sizes\cite{KKO}. By limiting the place where the light is emitted, energy is allowed to flow into the place autonomously by the movement of the dressed photons. If we assume that modeling as a quantum walk on a jellyfish graph is appropriate, these phenomena can be understood as a consequence of the theorem presented above, not only qualitatively but also quantitatively.

On the other hand, these phenomena cannot be understood as "random walk". This is because when we consider the "random walk" on a jellyfish graph, the probability in the limits is equal on all edges (The amount $J$ is, so to speak, the amount of deviation from "random walk" and it can be said that it measures a kind of "symmetry breaking"). Then, only the number of edges to which the vertices are connected, or "degree" is important, which contradicts the importance of dissipation. This also contradicts the fact that the density of the dressed photons is high at "vertices with rather small degrees", such as the endpoints. 
These phenomena suggest that the model of quantum walk on a jellyfish graph is appropriate for the analysis of the behavior of dressed photons.









\section*{Acknowledgment}
The authors thank Prof. Motoichi Ohtsu and Prof. Izumi Ojima for the fruitful discussions on dressed photons and quantum fields. They also thank Mr. Suguru Sangu and Mr. Etuso Segawa for the discussions on the modelling by quantum walks. This work was supported by Research Origin for Dressed Photon.

\nocite{*}
\bibliographystyle{eptcs}
\bibliography{generic}
\end{document}